\documentclass[11pt,reqno]{amsart}
\usepackage{amsmath, amssymb, enumerate, amsfonts, mathrsfs}
\usepackage{pgfplots}
\usepackage{color}
\usepackage{graphicx}
\pgfplotsset{compat=newest}
\usepackage{amsfonts,amscd}

\usepackage[top=30mm,bottom=30mm,left=27mm,right=27mm]{geometry}

\newtheorem{thm}{Theorem}[section]
\newtheorem*{theorem*}{Theorem}

\renewenvironment{proof}{\par\noindent{\bf Proof.}}{$\square$\par\bigskip}

\newtheorem*{lemma*}{Lemma}

\def\Re{\operatorname{Re}}

\def\exp{\operatorname{exp}}

\def\log{\operatorname{log}}

\def\diag{{\rm diag}}

\begin{document}

\date{\today}
\thanks{Research of the first author is partially supported by an NSERC Discovery grant 03927.\\}

\title{Markov processes and brain network hubs}

\author[M. Ram Murty]{M. Ram Murty}

\author[Asuri Narayan Prasad]{Asuri Narayan Prasad}

\address{Department of Mathematics and Statistics, Queen's University, Kingston, Ontario, K7L 3N6. Canada}
\email{murty@queensu.ca}

\address{Department of Clinical Neurological Sciences, Western University, London, Ontario, Canada}
\email{narayan.prasad@lhsc.on.ca}

\maketitle

\section{\bf Introduction}

The human brain (along with its emergent functions) seems to
be the 
final frontier of science.    This exploration  has been difficult largely due to an inadequacy of tools.  Now, in the 21st century, we have made some technological
progress towards the understanding of neural networks and in this endeavour, the mathematical branch of graph theory has
proved to be useful as a language to model the inner cerebral architecture.  
\par
In this paper, we propose that the firing of neurons or more precisely, {\bf neural assemblies}, can be modelled as a Markov process.
(We will review basic Markov theory in section 5.)
This idea is not new and has been proposed earlier by others (see for example, Chen et.al. \cite{chen} and Escola et. al. \cite{escola}).  The novelty of this paper is to
apply the main theorems of Markov chains and determine their implications in our understanding
of  brain networks and to recognize ``brain hubs''.  We will model these networks as directed graphs and so subtler aspects of Markov theory have
to be invoked.  It is not the adjacency matrix (or connectivity matrix as it is called in neuroscience) that plays
a central role, but rather what we call the Markov matrix of the graph where each connecting directed edge
carries a certain probability weight.  The spectral theory of these matrices plays a central role in determining
many of the graph's properties such as the diameter of the network (see Theorem \ref{main} below).
Markov theory seems ideally suited to study a neural network since it embodies two of  its fundamental properties:
recurrence and emergence.  This latter property is often described poetically as
``the whole is greater than the sum of the parts.''
\par
  At the end of our paper, we try to package the information of the neural network
into a ``zeta function'', which is technically a matrix of functions.  The study of this zeta function,
motivated by analogies of other zeta functions in mathematics, may be of value in the evolution of a mathematical
theory of neural networks.  As we have tried to make this paper accessible to both mathematicians and
neuroscientists, we will give a short historical overview from the neuroscience perspective.
\section{\bf Historical background to the concepts of neural networks}
Our understanding of brain function over the last two hundred years owes a lot to the concept of the neuronal doctrine,  
 extending from the work of Santiago Ram\'on y Cajal and others.
 This states that the structural and functional unit of the nervous system is the individual neuron \cite{hubel2}.
 Through a reconstructive process involving single neuronal cells and their connections, some aspects of the design logic in the structure and function of neural pathways in the nervous system were deciphered. 
 The excessive focus on the single neuron, its electrophysiology and behavior however, detracts from the ability to develop a deeper understanding that would account for brain function in behavioural or cognitive states associated with health and disease \cite{2}. The neural connectome may behave in ways that are dependent on the nature and complexity of the network that have been acquired in the evolutionary process from a microscale with few hundred neurons to several millions in the human brain.  {\bf Thus, it is important to conceptualize a neuronal hub or nodes in a microscale network as involving individual or a few neurons in a network to specialized neuronal assemblies of groups of neurons or even regions that go to form a connectome in the human brain connectome. We will elaborate on this a little further in the next section.}
\par
Over the last hundred years the frameworks have moved from the conceptual topographically organized receptive fields in cortical ‘columns’  described by Mountcastle, by Hubel and Wiesel  to the idea that the single neuron was not only the anatomical and functional unit of the brain but also its perceptual unit. However, this approach could not explain the findings that neurons could be engaged in intrinsic and emergent functions unrelated to sensory stimulus or motor action.
It was Cajal's disciple,  Rafael Lorente de N\'o,   who argued that the structural design of many parts of the nervous system is one of recurrent connectivity, where neuronal signals can reverberate among groups of neurons \cite{no}.
 This idea was further developed by Donald Hebb \cite{hebb}, who proposed  in 1949  the idea of ‘cell assemblies’ where neural circuits worked by sequentially activating groups of neurons where recurrent and reverberating patterns of neuronal activation occurring within these closed loops, would be responsible for generating the various states of brain function.  
This foundational work is often described as being on par with Darwin's {\sl Origin of the Species.}
This work presented a neurological basis for human behavior and what is called
Hebb's law that ``neurons that fire together, wire together'' is now a fundamental mantra of neuroscience.
This law has also led to the neuroplasticity thesis of the human brain.  In essence, the brain has the capacity to change itself.
Further modification of these ideas came about with the development of electroencephalography and the identification of spontaneous oscillatory activity in neurons.
\par
The concept of a neural network is further supported by anatomical evidence that neuronal assemblies give rise to a connectivity matrix which encodes communication between neurons or neuronal assemblies.  That is, 
 each neuron or neuronal assembly receives inputs from many other neurons or neuronal assemblies while sending its outputs to large populations of cells and this process is encoded in the matrix.  The term ‘neural network’ model is one made up of distributed neural circuits in which neurons or
neuronal assemblies are abstracted into nodes and linked by connections that change through learning rules.  
\par
Neural network assemblies involve two basic types: feedforward networks, which are governed by one-way connections and recurrent networks, in which feedback connectivity is dominant.  Rather than fixed structure assemblies, it is more likely that a third type of neural network can be conceptualized to operate within biological organ systems in more fluid dynamic states (that are stochastic, not deterministic) permitting the allocation of weights to be asymmetric and exhibiting transient dynamical patterns without stable states \cite{maass}.
 The asymmetry in the synaptic connectivity matrix naturally endows these models with temporally organized activity without necessarily requiring an input signal.  Within such networks, individual neurons acquire flexibility in participating in different networks at different times.  This combinatorial flexibility \cite{hebb},  a consequence of synaptic plasticity, permits the modular composition of small assemblies into larger ones, within which neural circuits that are constantly in flux give rise to and encode emergent properties such as time, providing different time stamps to events as they arise \cite{maass}.  
\par
This suggests that we apply continuous time Markov theory to model the network.  But then, this would require explicit knowledge
of the transition probability functions, as functions of time.  We indicate below how the main theorem of continuous time
Markov theory can be used.  Future research in neuroscience could determine the nature of these transition functions
that can then be applied to determine hubs using such a model.  
\par
An approach to the analysis of neural networks and connectivity matrices in the brain as visualized by current neurophysiological and neuroimaging technologies (EEG,  high density arrays, magnetoencephalography, tractography and anisotropy, functional neuroimaging) is a mathematical one using graph theory that has become central to the identification of network hubs. 
In this paper, we inject both discrete and continuous theories of Markov chains to study brain networks.  We also propose
that one can go in the reverse direction and reconsider large  tracts of classical graph theory from the brain network perspective
and derive new results essential to our understanding of these vast  networks within networks  (``worlds within worlds'') that reside in the human brain.
Theorem \ref{main}  below  is an example of this symbiotic theme where eigenvalues of the Markov matrix (and not the adjacency matrix as is customary in spectral graph theory) are related to the diameter and radius of the graph.  

\section{\bf Neural networks in disease states: epilepsy as a network disorder}

Data from clinical, neuropathology, and imaging studies suggest that there is a pervasive disorganization of neural networks in patients with epilepsy \cite{bernhardt-2019}.
 Imaging studies have shown atypical hub organization in epilepsy relative to controls and associations between hub topography and cognitive dysfunction, and have suggested utility in hub mapping for postsurgical outcome prediction \cite{carboni} \cite{lariviere}.
\par
The patterns of disrupted organizations within networks in patients with epilepsy are described for patients with both focal and generalized epilepsy \cite{hong}.
 These are influenced by many variables:   patient age, age at seizure onset, and disease duration.  
This is
preliminary,  but compelling evidence indicates that large scale network abnormalities lead to the complex effects in the form of cognitive impairment symptoms in epilepsy, particularly in executive functioning, semantic and retrograde memory, and naming.
\par
While a regional loss in connectivity has been demonstrated in focal epilepsy (temporal lobe epilepsy, TLE), there appears to be measurable perturbations in other parts of the neural networks that show increased clustering and shorter path lengths 
\cite{bonhila}.  Despite regional connectivity loss, it is observed that small-world properties of the network remained spared. Network analysis has been effectively used to predict treatment resistance, response to surgery, failure of surgery and the effects on cognition and behavior (pre and post-surgery) \cite{lariviere}.

\section{\bf Graph theory as applied to neural networks and the concept of hubs}

The entire neural network and its connections in the human brain has been termed the ``connectome''.  
 Here the concept of the connectome is operative on a macroscale.  Nodes represent distinct neural elements, such as specialized neuronal assemblies and edges represent connections between nodes. 
Developmental processes determine how connections within this network develop resulting in certain neural assemblies within regions possessing a large number of network connections which are putatively marked as ``network hubs''. 
 In the human brain visualization of the microscale is not possible even when imaging studies with high resolution used conventionally. Using methods of white matter tractography and resting-state fMRI it has been possible to construct functional and structural connectomes. These are quantitative representations of brain architecture as neural networks, comprised of nodes and edges. The connectomes, typically depicted as matrices or graphs, possess topological properties that inherently characterize the strength, efficiency, and organization of the connections between distinct brain regions.

These nodes and hubs functionally and on account of their location serve to direct a large fraction of signaling traffic within the brain and are implicated in many neurological and psychiatric disorders such as epilepsy, autism etc.  \cite{royer}  
 Such hubs are vulnerable to pathological insults and consequently develop structural alterations that enhance network excitability.  
\par
Multicellular and macroscale phenomenology 
generated by such network arrangements and rearrangements and the corresponding emergent properties leading to an understanding of brain functions in health is currently a focus of neural network research.  See for example, 
\cite{shine} and \cite{farahani}.  
\par
The application of graph theory to analysis of network connectivity matrices has already progressed despite the many limitations of current methodologies employed and the variability of results across different modalities employed in studying connectivity within the brain.
In this context, a {\sl connectome} is defined as a ``comprehensive structural description of the network of elements and
connections of a given nervous system.''  Connectomics \cite{sporns} is then the branch of neuroscience focused on reconstructing
the connectome of various organisms.  
 To study the structure of these networks with a view to understand and explain the divergent findings from imaging studies, connectivity, and network analysis, several authors \cite{bernhardt}\cite{royer}
have introduced the mathematical theory of graphs as a framework.
  The ambitious ``human connectome project'' \cite{sporns} consists of  mapping all the neural connections of the brain.
\par
Towards this goal, the concept of ``hubness'' and ``hubs'' in various contexts is the first step.  
For instance, different forms of connectivity of the brain regions, such as structural connectivity and functional connectivity are
described in \cite{bernhardt}.  Using this modeling, many authors have 
defined the notion of {\sl hubs} as being ``brain regions with high connectivity to other parts of the brain''
and ``situated along the brain's most efficient communication pathways.''  
The mathematics of graph theory was then proposed as a means to identify the hubs using various measures of {\sl centrality}.
For instance, {\sl degree centrality}, {\sl betweenness centrality} and {\sl eigenvector centrality} can be measured using
 graph theory.  Eigenvector centrality in particular arises by associating to the neural network a discrete time Markov
 process (similar to the Google PageRank algorithm) and identifying the hubs as those nodes with the highest probabilities
arising from the stationary distribution of irreducible Markov chains.  In this model, a transition probability is associated to
each pair of nodes $i,j$ and is equal to the reciprocal of the  number of outgoing edges from $i$ to $j$.   (We are assuming
at most one edge between any pair of nodes.  If there are multiple edges, one can define the transition probability as the ratio
of the number of directed edges from $i$ to $j$ divided by the total number of directed edges emanating from $i$.)
Keeping in mind that our nodes represent neural assemblies, 
all synaptic assemblies are not the same
and as their efficiency is determined by their synaptic strength \cite{dobrunz},  the discrete Markov model may not completely identify  all the hubs.
We therefore propose a continuous time Markov model, taking into account the synaptic strength (which can be measured using
various neuroimaging and electrophysiological modalities) and as one evolving in time so as to accommodate neural plasticity.  This gives rise to a continuous Markov process which can be studied using methods different from the discrete case
(described below).  
As hubs are often implicated in the generation of epileptic seizures and may be involved in other brain disorders due to their vulnerability, a better delineation of these neuronal assemblies may lead to a more precise surgical intervention and/or the identification of targets that are affected selectively in disease states and lead to early recognition of such pathology in association with neuro-psychiatric symptoms.

\begin{center}
\includegraphics[width=\textwidth]{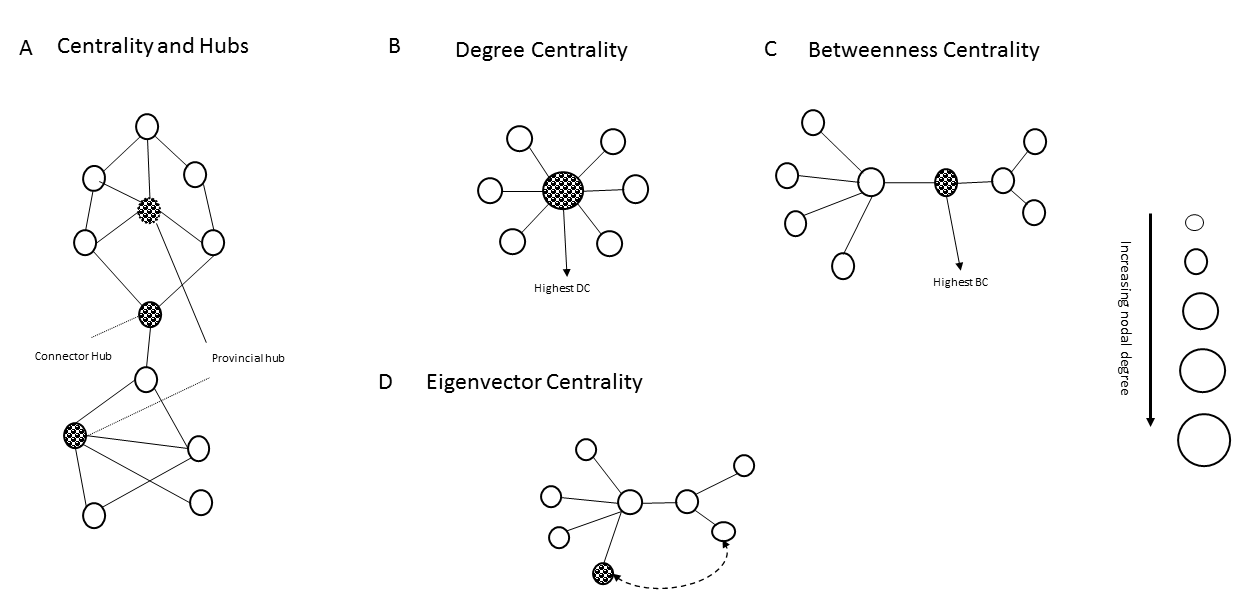}
Figure 1:  Various concepts of nodal centrality
\end{center}
\par $\quad $
\par
Let us briefly discuss degree centrality and betweenness centrality. (For basic notions of graph theory, the reader can refer to \cite{cm}.)
If we model the brain network using graph theory, the nodes (or vertices) of the graph $G$ are the brain regions, and the
directed edges represent connections between the regions as measured by various neuro-physiological measures and imaging
such as EEG, MEG, fMRI and MRI tractography.  
There is no loss of generality in assuming our graph $G$ is connected.  That is, given any two distinct nodes $x,y$, there is a directed path from $x$ to $y$.
In such a directed graph, one has the notion of a ``directed distance''
between two nodes $x,y$ as the length $d(x,y)$ of the shortest directed path from $x$ to $y$.  Following the usual terminology of
graph theory, the {\sl eccentricity} $e(x)$ of a node $x$ is simply
$$ e(x):= \max_{y\in G} d(x,y).  $$
The {\sl diameter} and {\sl radius} of the graph are the maximum and minimum nodal eccentricities respectively.  
One can use eccentricity as a means of identifying hubs and this leads to the notion of {\sl betweenness centrality}.  
\par
The {\sl outdegree} of a node $x$  is the number of edges emanating from $x$ and is denoted as $d^+(x)$.  The
{\sl indegree} is the number of edges entering into the node $x$  and is denoted $d^-(x)$.  If the graph is undirected,
$d^+(x)=d^-(x)$ and the {\sl degree} is then this common value.  Royer and others have modeled the brain network
as an undirected graph and introduced several measures of {\sl centrality} such as {\sl degree centrality} (nodes with the largest
degree),
{\sl betweenness centrality} (discussed above) and {\sl eigenvector centrality} (discussed below).  Each of these measures of centrality (see Figure 1) identifies hubs
in different ways \cite{royer}.    In this paper, we propose a new model of the brain network as a directed graph and a notion of centrality that emerges viewing this  network as both a discrete and continuous Markov process.  Hubs are then identified using these two perspectives.  

\section{\bf The discrete time Markov model}
In the simplest model discussed in the literature, two nodes $i$ and $j$ of the graph $G$ associated to the neural network of the brain are deemed adjacent if there is a directed edge from $i$ to $j$.  The discrete Markov process that is then associated to $G$ is
given by the matrix $P$ of transition probabilities $p_{ij}$ which is equal to the reciprocal of the outdegree of $i$.
If there is no directed edge from $i$ to $j$, the probability is set to zero. 
If there are multiple edges from $i$ to $j$, one can set the probability to be the total number of edges from $i$ to $j$ divided by
the outdegree of $i$.  
 It is convenient to call $P$ the Markov matrix associated to $G$.  Let us note that our matrix $P$ is {\sl row stochastic}, that is,
the sum of the entries in each row adds up to 1 since it represents the probability that the process moves from state $i$ to some other state.
(One can consider a more general situation where $P$ is
the matrix of transition probabilities which are not necessarily determined by the indegrees and outdegrees.)
\par
Our hypothesis that $G$ is strongly connected is equivalent to saying that we have 
what is called
an {\bf irreducible Markov chain}.  As the number of
neurons (or neuronal assemblies)  is finite, we conclude that all nodes are recurrent and non-null, in the sense of Markov theory (for
definitions of these terms, see for example, the lemma on page 225 of \cite{grimmett}).
In this setting, the fundamental theorem of Markov chains (see page 227 of \cite{grimmett}) applies and there is a unique stationary distribution $\bf v$
satisfying 
$$ {\bf v}P = {\bf v}.$$
This reduces the problem of determining hubs into a problem of linear algebra which is easily solved.
In other words,  the $i$-th coordinate of $\bf v$ encodes the probability that the process converges to $i$ in the limit.
Thus, hubs  can be identified as those nodes with the highest probabilities encoded in $\bf v$.  This is eigenvector centrality in a nutshell (and it is the same theoretical foundation for the Google PageRank Algorithm).  
\section{\bf The continuous time Markov model}
In the continuous model, we associate probability functions $p_{ij}(t)$ as a function of time $t$.  These could be functions giving the synaptic strength between
the neural assemblies indicated by $i$ and $j$.  
We thus get a matrix $P(t)$ whose entries are functions of time $t$.  
We will assume that our matrix entries are differentiable functions of $t$.
By the fundamental theorem of continuous time Markov chains, 
one can associate a 
matrix $Q$ (called the  generator matrix) such that 
$$P(t) = \exp( tQ), $$
and the stationary distribution $\bf v$ of $P(t)$ is determined by solving ${\bf v}Q=0$ (see pages 259-261 of  \cite{grimmett}).
This is again a problem of linear algebra.
The matrix $Q$ has a simple description as a matrix of {\sl transition intensities}:
$$q_{ii} = - p_{ii}'(0) = \lim_{t\to 0} {1-p_{ii}(t) \over t}, \qquad q_{ij} = q_{ii}p_{ij}'(0), \quad i\neq j. $$
Continuous time Markov theory then implies that ${\bf v}P(t)= {\bf v}$.  
In other words, the hubs can again be identified as those coordinates with the highest probabilities encoded by $\bf v$.  
This stationary distribution will (in general) be different than the one given by a discrete time Markov model.

\section{\bf The diameter of the brain network}

In this section, we will discuss the discrete case and the role the matrix $P$  plays as it relates to questions of classical graph theory.
In spectral graph theory, bounds for the diameter of a graph are often derived using eigenvalues of the adjacency matrix.
The above analysis, at least in the context of the brain network, suggests that it may be relevant to shift the focus and relate
the diameter to the eigenvalues of the Markov matrix $P$. As $P$ is a matrix of probabilities, it is a non-negative matrix
and the celebrated Perron-Frobenius theorem (which we state below) applies.   
As the sum of the entries of each row of $P$ equals  1, we see that the column vector {\bf J} consisting of all 1's is a (right)  eigenvector of $P$ with eigenvalue 1.  For an irreducible, aperiodic chain, the Perron-Frobenius theorem (see page 240 of
 \cite{grimmett})
implies that all the other eigenvalues have absolute value strictly less than 1.  We call these 
the non-trivial eigenvalues and define 
$$\rho := \text{ the maximum absolute value of the non-trivial eigenvalues of  } P.$$  
\begin{thm}[Perron-Frobenius]\label{pf}  Suppose that $P$ is a non-negative matrix which is irreducible.  Then $P$ has a real, positive eigenvalue
$\lambda_1$ with the following properties. 
\begin{enumerate}[{\rm (a)}] 
\item  Corresponding to $\lambda_1$, there is an eigenvector {\bf x}  all of
whose elements may be taken as positive.  
\item  Any other eigenvalue $\alpha$ satisfies $|\alpha|\leq \lambda_1$.
\item  $\lambda_1$ increases when any element of $P$ increases.  
\item  $\lambda_1$ is a simple eigenvalue of $P$.  
\item  $\lambda_1\leq \max_j \left( \sum_{k} a_{jk} \right), \qquad \lambda_1 \leq \max_k \left( \sum_{j} a_{jk} \right). $
\item  If $P$ has exactly $t$ eigenvalues of maximum modulus $\lambda_1$, then there is a non-singular matrix $C$ such that
$$ C^{-1}PC = \left( \begin{matrix} \Lambda & 0 \\ 0 & J \end{matrix}\right) $$
where $\Lambda$ is a $t\times t$ diagonal matrix ${\rm diag}(\lambda_1, \lambda_1 \zeta_t, ..., \lambda_1 \zeta_t^{t-1})$
with $\zeta_t=e^{2\pi i/t}$ and $J$ is a Jordan matrix whose diagonal elements are all strictly less than $\lambda_1$ in modulus.
\item  If $P$ is aperiodic, then $t=1$.
\end{enumerate}
\end{thm}

Before proving our theorem, we 
recall the notion of the {\bf norm} of a matrix $A$.  
This is defined as 
$$ ||A||:= \max_{x\neq 0} {||Ax|| \over ||x||},$$
where $||x||$ denotes the usual Euclidean norm of a vector.
If $A^T$ is the transpose of $A$, then it is easy to see that $||A||=||A^T||$.  
The {\bf condition number},
denoted $\kappa(A)$ of
an invertible matrix $A$ is defined as
$$\kappa(A):= ||A||  \, ||A^{-1}||$$
where $||A||$ is the norm of the matrix $A$.  The condition number of $A$ is invariant
under permutation of its columns and re-scaling them.  If $A$ is any
non-singular $n\times n$ matrix with real entries, then $A^{T}A$ is 
symmetric and positive definite.  
Therefore, the eigenvalues of $A^TA$ are all real and positive.  
The {\bf singular values} of $A$
are then defined to be the positive square roots of the eigenvalues of
$A^TA$ and we can order them as $\sigma_1 \geq \sigma_2 \geq \cdots 
\geq \sigma_n >0$.  The condition number of $A$ can be shown to be equal
to $\sigma_1/\sigma_n$ (see for example pages 280-286 of \cite{strang}).

With this interlude, we now apply the Perron-Frobenius theorem to prove:

\begin{thm}  \label{main} Let $G$ be a connected graph with Markov matrix $P$ which is aperiodic
with simple eigenvalues.    Let {\bf v} be its
unique stationary distribution with components $v_j$.  Then $v_j>0$ for every $j$.
If $C$ is the matrix whose columns are a basis of eigenvectors of $P$ and 
$\kappa(C)$ is the conditon number of $C$, then
 eccentricity of any node $i$  is bounded by $\max_j (\log v_j)/(\log \rho)$.  Thus, the diameter and radius of $G$ are bounded by
 $$ \max_j {\log [\kappa(C)/v_j] \over [\log 1/\rho]} .$$
\end{thm}
\begin{proof}  Let $G$ have $M$ nodes.  Since $G$ is connected, our matrix $P$ is irreducible.  The aperiodicity of $P$
implies that $t=1$ in the Perron-Frobenius Theorem \ref{pf}.  As all the eigenvalues are simple, there is a basis of eigenvectors.
As $P$ is row stochastic, it  has eigenvalue 1 with corresponding eigenvector {\bf J}.  By
(e), this must be the largest eigenvalue and by
  (b), all the other eigenvalues are
strictly less than 1 in absolute value.  Thus, $\lambda_1=1$.  
Let $C$ be the matrix whose columns are a basis of eigenvectors of $P$ ordered so that the first column is {\bf J}.  
 Thus, $PC=C\diag(\lambda_1, ..., \lambda_M)$ so that 
\begin{equation}\label{diag}
C^{-1}PC = \diag(\lambda_1, ...,\lambda_M), 
\end{equation}
where $\lambda_1 =1$ and $\lambda_2, ..., \lambda_M$ are the non-trivial eigenvalues.  
A routine induction argument shows
$$C^{-1} P^n C = \diag(\lambda_1^n, ..., \lambda_M^n)$$
and consequently
$$P^n = C \diag(\lambda_1^n, ..., \lambda_M^n) C^{-1}. $$
Writing $P^n = [p_{ij}^{(n)}]$, 
$C= [c_{ij}]$ and $C^{-1}= [b_{ij}]$, we have by the rules of matrix multiplication that
\begin{equation}\label{limit}
p_{ij}^{(n)} = \sum_{k=1}^M c_{ik} \lambda_k^n \, b_{kj}. 
\end{equation}
It is easy to see that
\begin{equation}\label{cauchy}
 \sum_{j=1}^M |c_{ij}|^2 \leq ||C||^{2} , \qquad \sum_{k=1}^M |b_{kj}|^2 \leq ||C^{-1}||^{2}\qquad \forall \,  i,j. 
\end{equation}
Indeed, the left hand side of both inequalities in (\ref{cauchy})
can be viewed as the norms of the vector $e_iC$ and $Be_j$ where
$e_i$ and $e_j$ denote the standard basis vectors with zero entries except in the $i$-th position and $j$-th position (respectively) where it is equal to 1.  Then (\ref{cauchy}) is immediate
from the definition of the norm and the observation that $||C^T||=||C||$. 

Noting that $\lambda_1=1$ and 
all the other non-trivial eigenvalues have absolute value strictly less than 1, we see from (\ref{limit}) that
$$ \lim_{n\to \infty} p_{ij}^{(n)} = c_{i1}b_{1j}. $$
Recalling that the first column of $C$ is the eigenvector {\bf J} (rescaled), we see that $c_{i1}$ does not depend on $i$
and is a fixed scalar.  
We can therefore write $v_j=c_{i1}b_{1j}$.
Thus, we see from (\ref{limit}) that
\begin{equation*}
p_{ij}^{(n)} = v_j + \sum_{k=2}^M c_{ik} \lambda_k^n b_{kj}.  
\end{equation*}
Applying the Cauchy-Schwarz inequality to the sum on the right and using (\ref{cauchy}), we deduce
\begin{equation}\label{final}
|p_{ij}^{(n)} - v_j | \leq \kappa(C)\rho^n.  
\end{equation}
The components of the vector ${\bf v} = (v_1, ..., v_M)$ are the stationary probabilities of the process and so the components are
non-negative and add up to 1.  
The Perron-Frobenius theory however implies that the $v_j$'s are all strictly positive.  This can be seen very easily  without 
appealing to the theory as follows.  
From (\ref{diag}), we have $C^{-1}P=\diag(\lambda_1, ..., \lambda_M) C^{-1}$ and the first row of $C^{-1}$
is the vector $(b_{1j})$.  Keeping in mind that $\lambda_1=1$, the equation $C^{-1}P=\diag(\lambda_1, ..., \lambda_M) C^{-1}$
implies that the row vector ${\bf v} = (v_1, ..., v_M)$ is a left eigenvector of $P$ with eigenvalue 1.  Hence, ${\bf v}P={\bf v}$
and by induction ${\bf v}P^n = {\bf v}$ for all $n\geq 1$.  
As the sum of the $v_i$'s equals 1, there is a $k_0$ such that $v_{k_0}>0$.   From ${\bf v} = {\bf v}P^n$, we have
\begin{equation}\label{keys}
 v_j = \sum_{k=1}^M v_k  p_{kj}^{(n)} \geq v_{k_0}p_{k_0 j}^{(n)}
\end{equation}
for any $n$.   As our graph is connected,  there is an $n_0$ such that $p_{k_0j}^{(n_0)} >0$. Choosing $n=n_0$ in (\ref{keys}), we deduce $v_j>0$.  
Returning to (\ref{final}), we see that
$p_{ij}^{(n)}>0$ if $v_j - \kappa(C)\rho^n >0$ which is the case if $1/v_j < \kappa(C)^{-1}(1/\rho)^n$.  In other words, if 
$n > (\log \kappa(C)/v_j)/\log (1/\rho)$, we have $p_{ij}^{(n)}>0$, which means there is a path of length $n$ from $i$ to $j$.
Since the eccentricity of a node $x$  is the maximum of $d(x,y)$ as $y$ ranges over the nodes of the graph $G$,
the eccentricity is bounded by 
$$ \max_j {\log [\kappa(C)/v_j] \over [\log 1/\rho]} .$$
As the diameter and radius are the maximum and minimum of the nodal eccentricities, the last assertion is immediate.
This completes the proof.
\end{proof}
From  the expression of Theorem \ref{main} as
$$ {\log \kappa(C)/v_j\over \log 1/\rho }, $$
we see
that to minimize the eccentricity and thus bound the radius, one needs first  to maximize the denominator and so minimize $\rho$.
We then need to minimize the numerator and this means to maximize $v_j$.  In other words,
the ``leader hub'' can be used to bound the radius.
\par
Another remark is that if $C$ is an orthogonal matrix, which is the case
if $P$ is symmetric, then $\kappa(C)$ is equal to 1.  It is possible to
replace $\kappa(C)$ with a ``spectral bound'' depending only on the eigenvalues
of $P$.  This work involves another approach and is the subject matter
of the forthcoming paper \cite{carter-murty}.
\par
The problem of classification of graphs in which $\rho$ is minimized
is a new problem of graph theory and has not been studied before except in some very special cases like regular graphs
in which the degree of every node is the same.  But in this case, the adjacency matrix and the Markov matrix are essentially
the same apart from a scaling factor.  Thus, the study of brain networks using graph theory leads to some new and
interesting problems in mathematics hitherto unexamined.
\par
One additional comment is worth noting.  In our theorem, we assumed that all the eigenvalues were simple.
A random matrix has simple eigenvalues and so generically, this is not  a stringent hypothesis.  One can
treat the multiple eigenvalue case but the details are a bit more involved.
\par
\section{\bf Application of the model to the nematode \textbf{\textsl{ Caenorhabditis    elegans}}}
As a preliminary application of the theory, we applied our discrete Markov model  to the connectivity matrix of the neural network of the microscopic roundworm {\sl C elegans} which is available
from the Dynamic Connectome Lab\footnote{https://www.dynamic-connectome.org/resources}. 
This is one of the few organisms for which a complete description of the neural network is known (see page 82 of \cite{sporns}).
The 277 $\times$ 277 matrix is
easily converted  into a Markov matrix as indicated above.  There are only two non-zero stationary probabilities, each with
value very close to $1/2$  corresponding
to the PVCL and PVCR neurons whose functions are related to forward locomotion and responding to any harsh touch to the tail
respectively.  These are the hubs of {\sl C elegans}.  
\par
In a recent computational study using the published connectome of the C elegans hermaphrodite as a combined (i.e., chemical synapses and gap junctions), directed, and weighted network identified 12 critical neurons out of 279 neurons in the network using a vulnerability analysis methodology.  
These 12 critical neurons include the two  that we identified using the associated Markov matrix algorithm.
 These two are interneurons; they control elements that display a high vulnerability score, i.e their loss has a significant effect on global efficiency in the computerized analysis conducted in the study \cite{kim}.
An entire issue of the Proceedings of the Royal Society (Biological Sciences) was dedicated to the neuroscience related to this organism \cite{larson}.
\par
We have also explored the applicability of the proposed approach to the neural network connectome of Drosophila
(fruitfly) \cite{winding} to identify critical neuronal constituents. The findings are the subject of our next paper to demonstrate the applicability of the proposed mathematical approach using the Markov theory to more complex connectomes than the C elegans connectome.

\section{\bf The zeta function of a finite network }
The quantity $\rho$, defined as the maximum absolute value of the non-trivial eigenvalues of $P$ played an important
role in Theorem \ref{main}.  This is related to the ``spectral gap'' phenomenon well-known in the theory of Ramanujan graphs and expander graphs
as well as the theory of zeta functions in number theory. 
The concept of a zeta function to study complex phenomenon has been a fundamental idea in number theory, algebraic geometry and graph theory.  Inspired by these analogies, 
a  similar concept can be introduced in the context of neural networks.  Indeed, if given any (finite) directed graph $G$,
with associated Markov matrix $P$, 
we define the zeta function of $G$ as the matrix function
$$Z_G(t) := \exp\left( \sum_{n=1}^\infty {P^n t^n \over n} \right) . $$
This resembles many zeta functions that appear in algebraic geometry and related areas.  But let us note that
our zeta function above is a matrix function.  The proof of 
Theorem \ref{main}  shows that this is a matrix of rational functions.  Unlike spectral graph theory which relates
eigenvalues of the adjacency matrix to properties of the graph, the zeta function $Z_G(t)$ relates eigenvalues of 
the Markov matrix $P$ to hubs of the network.  We already saw this phenomenon in Theorem \ref{main}.
\par
It may also be useful to consider the ``dual graph'' $G^*$ of the directed graph $G$ where now the nodes are the same
but the direction of each edge is reversed.  The associated Markov matrix is then denoted $P^*$
with $\rho^*$ being the maximum of the absolute values of the non-trivial eigenvalues of $P^*$.
   The spectral analysis of $P^*$ 
can be carried out as we did with $P$ and one can relate $\rho^*$ to the study of the diameter of $G$.
The stationary vector associated with $P^*$ encodes the probabilities of the process starting at any given node.  
\par
It has also been found that the theory of random graphs is not applicable to understanding the neural network of the brain
\cite{albert}.  For instance, in the case of random graphs, the degree distribution is a Poisson law.
By contrast, the neural network of the brain seems to follow a power law.  Empirically, it looks like there is a constant
$s$ such that $P(X=k) \asymp k^{-s}$, where $X(v)$ is the random variable giving the degree of a node $v$.
What this suggests is that the degree of random node follows a zeta distribution in that 
$$P(X=k) = \zeta(s)^{-1}k^{-s} $$
for some suitable $s$ with $\zeta(s)$ denoting the Riemann zeta function.   This signals a need for a  careful study of graphs whose degree distribution follows a power law.  In particular, if $a_k$ is a sequence of non-negative real numbers and
$$F(s) := \sum_{k=1}^\infty {a_k \over k^s} $$
is its associated Dirichlet series which converges for $\Re(s)>1$, then the study of random variables $X$ whose
distribution is given by
$$P(X=k) = F(s)^{-1} {a_k \over k^{s}}, $$
seems relevant.    These are  topics for future research.
\section{\bf Discussion}
We have presented a new model to identify brain network hubs that may be of use in the deeper understanding and treatment of epilepsy
and other disorders.    As various
modalities of measurement are available, it should be possible to clinically test or at least simulate the validity of this hypothesis.
The nature of functions describing synaptic strength is studied in \cite{trensch}.  These may be relevant in our continuous time Markov model.  
Graph theory can thus be used in these studies in a fundamental way.  In this paper, we highlighted the importance of Markov theory,
but it is becoming increasingly clear that other deeper chapters of mathematics can be injected into the study of the brain, most
notably, the theory of expander graphs.  
By all descriptions, the wiring of the brain connectome is not random nor regular but ``rather characteristic of a small world'' \cite{bernhardt}.  The network exhibits properties of a subclass of graphs called {\sl expander graphs} which have few edges and yet have high connectivity,  ``an
architecture that enables both the specialization and the integration of
information transfer at relatively low wiring costs''  \cite{bernhardt}.
On the other hand, our Theorem \ref{main}   suggests that brain network theory gives rise to new questions in 
 graph theory that have not been studied before.
 This opens up a new symbiosis between neuroscience and graph theory.
$$\quad $$
\noindent {\sl Acknowledgements.}  We thank Professors Garima Shukla (Division of Neurology, Faculty of Medicine, Queen's University, Kingston, ON, Canada),
Siddhi Pathak (Chennai Mathematical Institute), Seoyoung Kim (Postdoctoral Fellow, Mathematisches Institut
Georg-August Universit\"at G\"ottingen), Sunil Naik (Postdoctoral Fellow, Queen's University),
 Rebecca Carter (Graduate student, Queen's University)  for their helpful comments and Nic Fellini (Graduate student, Queen's University) for the linear algebra verification related to the connectivity matrix
of {\sl C elegans}. 
$$\quad $$

\end{document}